\documentclass[a4paper, 11pt]{article}
\usepackage[english]{babel}
\usepackage[fleqn]{amsmath}
\usepackage{color} 					 
\usepackage{relsize}
\usepackage{listings}
\usepackage{mathtools}
 \usepackage{amsthm}

\usepackage{enumerate}
\usepackage{placeins}
\usepackage{scrextend}
\addtokomafont{labelinglabel}{\bfseries\sffamily}

\usepackage[all,cmtip]{xypic} 
\usepackage[pdftex]{graphicx}
\usepackage[font=small,labelfont=bf]{caption}
\usepackage{subcaption}
\usepackage[pdftex]{thumbpdf}
\usepackage{pdflscape}
\usepackage{tabularx}
\usepackage{amssymb}
\usepackage{amstext}
\usepackage{mathrsfs}
\usepackage{verbatim}
\usepackage{colonequals}
\usepackage{amsthm}
\usepackage{thmtools}
\usepackage{geometry}
\usepackage{hyperref}

\usepackage{todonotes}

\usepackage{tikz}
\usetikzlibrary{shapes,arrows}
\tikzstyle{block1} = [rectangle, draw, fill=blue!20, 
text width=10.5em, text centered, rounded corners, minimum height=2em, minimum width=10.5em]
\tikzstyle{block2} = [rectangle, draw, fill=red!20, 
text width=10.5em, text centered, rounded corners, minimum height=2em, minimum width=10.5em]
\tikzstyle{line} = [draw, -latex']

\newcommand\subsetsim{\mathrel{%
  \ooalign{\raise0.2ex\hbox{$\subset$}\cr\hidewidth\raise-0.8ex\hbox{\scalebox{0.9}{$\sim$}}\hidewidth\cr}}}
  \newcommand\supsetsim{\mathrel{%
  \ooalign{\raise0.2ex\hbox{$\supset$}\cr\hidewidth\raise-0.8ex\hbox{\scalebox{0.9}{$\sim$}}\hidewidth\cr}}}

\def\RR{{\mathbb{R}}}

\def\ZZ{{\mathbb{Z}}}

\def\diam{\mathrm{diam}\,}

\def\orad{{\texttt{rad}}}
\def\parent{{\uparrow \hspace{-.1em}}} %
\def\grandparent{{\upuparrows \hspace{-.1em}}}

\DeclareMathOperator*{\argmin}{argmin}
\DeclareMathOperator*{\rank}{rank}

\theoremstyle{plain}
\newtheorem{theorem}{Theorem}[section]
\newtheorem*{example}{Example}
\newtheorem{assumption}{A}

\declaretheorem[name=Lemma, style=plain]{lemma}
\declaretheorem[name=Definition, style=definition, sibling=lemma]{definition}

\DeclareMathOperator{\ran}{ran}

\newcommand\vspan{\mathrm{span}}

\author{Bernhard~Brehm\\
Hanne~Hardering}
\title{Sparips}
\date{\today}
\begin{document}
\maketitle
%
%
%
%
\abstract{
Persistent homology of the Rips filtration allows to track topological features of a point cloud over scales, and is a foundational tool of topological data analysis. 
Unfortunately, the Rips-filtration is exponentially sized, when considered as a filtered simplicial complex. Hence, the computation of full persistence modules is impossible for all but the tiniest of datasets; when truncating the dimension of topological features, the situation becomes slightly less intractable, but still daunting for medium-sized datasets.

It is theoretically possible to approximate the Rips-filtration by a much smaller and sparser, linear-sized simplicial complexs, however, possibly due to the complexity of existing approaches, we are not aware of any existing implementation.

We propose a different sparsification scheme, based on cover-trees, that is easy to implement, while giving similar guarantees on the computational scaling. We further propose a visualization that is adapted to approximate persistence diagrams, by incorporating a variant of error bars and keeping track of all approximation guarantees, explicitly.
}

\tableofcontents
\section{Introduction}
In topological data analysis a common goal is to extract topological information from a finite set of sample points $X=\{x_0,\ldots,x_n\}$ with a symmetric distance function $d:X\times X\to [0,\infty]$. Often one assumes that $d$ is a metric; that is, that $0<d(x,y)<\infty$ for each $x\neq y$, and that the triangle inequality holds. While the sample points can typically be considered as elements of some $\RR^k$, we will instead consider $(X,d)$ as a weighted graph, or a finite metric space (if $d$ happens to be a metric).

Naively, a finite topological space would just be a finite set of disconnected points.
A widely used tool to add topological structure to such a weighted graph $(X,d)$
is a family of abstract simplicial complexes called Rips-filtration. For every $r\ge 0$ it is given by $V(r)=\{M\subseteq X: \mathrm{diam}\,M< r\}\subseteq 2^X$, and has the property that $V(r)\subseteq V(s)$ for $r\le s$.
This complex is, by construction, a flag complex; that is, each simplex $\Sigma\subseteq 2^X$ is included in $V(r)$ if and only if every edge $(x,y)\subseteq \Sigma$ is included in $V(r)$. Alternatively, we can consider it as the clique-complex of the graph with all edges with weight less than $r$, where edges with weight $d(x,y)=\infty$ are missing. In an abuse of notation, we will not distinguish between finite flag-complexes and graphs, nor between filtered flag-complexes and weighted graphs.

Note that we could have permitted equality in the definition of $V(r)$ instead; which intervals contain which of their endpoints is an endless source of subtleties that are mostly irrelevant in practice. We encourage readers to ignore these subtleties, but will strive to get them right.

In order to obtain information about the topology of the underlying metric space, one studies the
\emph{persistence module}. This is a finite collection of finite-dimensional vector spaces $\{V_s\}_s$ over a field $\mathbb F$ with commuting linear maps $\phi^{t,s}:V_s\to V_t$ for $s\le t$, over a totally ordered set.
In particular, we consider the homology vector spaces $H(r)=\left( H_k(\mathbb{F},V(r)) \right)_{k=0,1,\ldots}$ that partially describe the topology of $X$ at scale $r$ and dimension $k$.
Since $V(r)\subseteq V(s)$ for $r\le s$, the embeddings induce linear maps $\phi^{t,s}=\iota_*^{t,s}:H(s)\to H(t)$.
An element in the range $\xi\in \ran \iota_*^{t,s}$ corresponds to a topological feature that persists (at least) from scale $s$ to $t$. It is intuitively obvious that long-lived features ($s\ll t$), are more interesting for understanding the structure of $X$ than short-lived features $(s\approx t)$.
Even though the collection $\{V(r):r\in\RR\}$ is formally infinite, it is ``finite in practice'', as for finite weighted graphs it decomposes into a finite number of intervals with identical complexes in each interval.
Most of our theory will concern such ``practically finite'' collections of finite-dimensional vector spaces $\{H(r)\}$.
For a survey of persistent homology see \cite{carlsson2009topology} and the references therein.

The computation of the persistence diagram typically scales badly with the number $N=|X|$ of sample points: The distance matrix will have $N(N-1)/2 \sim N^2$ entries, each of which needs to be computed and potentially stored in memory. The number of simplices of dimension $k$ scales like $N^{k+1}$ for $k\ll N$, and there are $2^N$ simplices in total.

A popular strategy to reduce the number of simplices is subsampling and thresholding: We first do a subsampling step, i.e., we take a subset of sample points that is dense enough in $X$ with respect to some parameter $\epsilon$, and then consider only the corresponding persistence diagram in the range $r\in [0,T]$ for some threshold $T$. In order to obtain proper performance, the scale quotient $Q=T/\epsilon$ needs to be bounded.

The procedure of subsampling and thresholding allows to resolve certain scales, slightly smaller than $T$, in the persistence diagram relatively precisely, up to a relative error depending on $Q$, without needing too high numbers of simplices. It does not allow us to resolve relatively long-lasting features. That is, we will find these features but cannot find out how long they last; we cannot use this approach to find out whether features present at $t_1 \ll t_2$ are actually the same. 
In this work, we address these shortcomings by the construction of a sparsified approximate Rips-complex.
Compared to the general pipeline for topological data analysis, our approach looks somewhat like the following:

\bigskip
\begin{centering}
	\begin{tikzpicture}[node distance = 5em, auto]
	\node [block1] (data1) {Raw Data};
	\node [block1, below of=data1] (sample) {Selected Subsample};
	\node [block1, below of=sample] (dist1) {Distance Matrix};
	\node [block1, below of=dist1] (module1) {Persistence Module};
	\node [block1, below of=module1] (intepret1) {Interpretation};
	\node [block1, right= 5em of data1] (data2) {Raw Data};
	\node [block2, below of=data2] (tree) {Contraction Tree};
	\node [block2, below of=tree] (dist2) {Sparse Distance Matrix};
	\node [block1, below of=dist2] (module2) {Approximate Persistence Module};
	\node [block1, below of=module2] (intepret2) {Interpretation};
	\path [line] (data1) -- (sample);
	\path [line] (sample) -- (dist1);
	\path [line] (dist1) -- (module1);
	\path [line] (module1) -- (intepret1);
	\path [line] (data2) -- (tree);
	\path [line] (tree) -- (dist2);
	\path [line] (dist2) -- (module2);
	\path [line] (module2) -- (intepret2);
	\node [below left= 1.5em and -4.5em of data1, align=right, anchor=east] (subsam)  {subsampling};
	\node [below left= 1.5em and -4.5em of sample, align=right, anchor=east] (distcomp) {evaluate ($O(N^2)$)};
	\node [below left= 1.5em and -4.5em of dist1, align=right, anchor=east] (assemble) {assemble and solve \\ (e.g. via \texttt{ripser})};
	\node [below left= 1.5em and -4.5em of module1, align=right, anchor=east] (visual) {e.g. visualization};
	\node [below right= 1.5em and -4.5em of data2, align=left, anchor=west] {reorganize \\ ($O(N\log N)$)};
	\node [below right= 1.5em and -4.5em of tree, align=left, anchor=west] {evaluate ($O(\epsilon^{-c}N)$)};
	\node [below right= 1.5em and -4.5em of dist2, align=left, anchor=west] {assemble and solve \\ (e.g. via \texttt{ripser})};
	\node [below right= 1.3em and -4.5em of module2, align=left, anchor=west] {e.g. visualization};
	\end{tikzpicture}
\end{centering}

%
The groundbreaking work of D. Sheehy \cite{sheehy2013linear} proposes one way of obtaining a sparsified complex for metric spaces (weighted graphs that fulfill the triangle inequality):
One constructs a length function $\ell$ that is $\infty$ for most pairs of points, and is therefore sparse (and obviously breaks the triangle inequality). This gives rise to a Rips-complex that is multiplicatively interleaved with the original one. This sparsified complex has, for large finite spaces $X$ of bounded doubling dimension $d$, asymptotically only linearly many simplices in $N=|X|$ (with, just like in subsampling and thresholding, bad scaling in the other parameters).
Further, it is shown in \cite{sheehy2013linear} that the sparsified complex can be constructed in asymptotic $N\log N$ time (and especially only $N\log N$ many evaluations of the distance function). \cite{sheehy2013linear} does not consider the interplay with subsampling.

From this, and the publication date of \cite{sheehy2013linear}, one might assume that our work is already done, and sparsified complexes and relative-error approximations are today standard in the field of applied topological data analysis. This is not the case, see e.g.\cite{otter2017roadmap}. There exists, to the best of our knowledge, no implementation of \cite{sheehy2013linear}, at all. 
We believe that this is in part due to the construction in \cite{sheehy2013linear} being excessively complicated.

In this paper we will develop and benchmark a \emph{practical} sparsification scheme. We made great efforts to ensure a construction that is easy to implement in a self-contained way, and easy to understand and modify. On the other hand, it is not optimized for provable asymptotic worst-case complexity, nor for clever time-saving tricks. However, the computational resources spent on sparsification are negligible compared to the computation of persistent homology of the sparsified complex, and many sample sets will be subsampled to a much smaller subset anyway. 

We implemented this sparsification scheme in the \texttt{julia} language, and will make it avaiable soon. 
We believe that our scheme should outperform the proposed schemes \cite{sheehy2013linear} in practice, but we cannot put this to the test, lacking any reference implementation of alternative approaches.

\section{Persistence modules and diagrams}

The persistence module of the Rips complexes $V(c)=\{M\subseteq X: \mathrm{diam}\,M< c\}$ is given by
\begin{align*}
H(c_{0}) \xrightarrow[]{\phi^{c_1,c_0}} H(c_1) \xrightarrow[]{\phi^{c_2,c_1}} H(c_2)  \xrightarrow[]{} \ldots ,
\end{align*}
where $H(c)=\{H_{k}(c)\}_{k=0,1,\ldots}$ consists of the homology vector spaces $H_{k}(c)=H_k(\mathbb{F},V(c))$, and $\phi^{t,s}=\iota_*^{t,s}$ is given by the embeddings $\iota^{t,s}: V(s)\hookrightarrow V(t)$ for $s\leq t$.

The following well-known classification theorem for persistence modules is the foundation of most of the field of topological data analysis: 

\begin{theorem}\label{T:classification}
	Let $\{V(c)\}_{c\in\mathbb{Z}}$ be a sequence of finite-dimensional $\mathbb{F}$ vector spaces,
	such that only finitely many $V(c)\neq \{0\}$, and let 
	$\phi^{c+1,c}: V(c)\to V(c+1)$ be linear maps.
	We set $\phi^{d,b}=\phi^{d,d-1}\circ\ldots\circ \phi^{b+1,b}$ for $b<d$, and $\phi^{b,b}=\mathrm{id}$. 
	
	Then there exists a collection of numbers $N_{d,b}\in\{0,1,\ldots\}$ and vectors $\xi^{d,b}_{c,i}\in V(d)$, for $b< c\le d$ and $i=1,\ldots, N_{d,b}$ such that:
	\begin{enumerate}
		\item for $b< c \le c' \le d$, we have $\phi^{c',c}\xi^{d,b}_{c,i}=\xi^{d,b}_{c',i}$, and
		\item for $b< c \le d <c' $, we have $\phi^{c',c}\xi^{d,b}_{c,i}=0$, and 
		\item the collection $\{\xi^{d,b'}_{c,i}:\, b'< b,\,d\ge c,\,i=1,\ldots,N_{d,b'}\}$ forms a basis for $\ran\phi^{c,b}$.
	\end{enumerate}
\end{theorem}	
For the sake of self-containedness, we give a proof of the classification theorem in Appendix \ref{app:classify}; alternatively see the references in \cite{carlsson2009topology}.

In the notation of the theorem, $(b,d]$ is called the life-time of a vector $\xi^{d,b}_{c,i}$; it is said to persist from birth (exclusive, before) $b$ to death (inclusive, after) $d$. A death at $d=\infty$ formally corresponds to vectors that are never in the kernel (for metric spaces this will be a single vector in $H_0$, corresponding to the fact that we always have a connected component). Likewise, a birth at $b=-\infty$ formally corresponds to vectors that are present in the very first space; for metric spaces, we get one such vector in $H_0$ for each distinct point. 

Persistence modules are (up to changes of coordinates) determined by the ranks  $r^{t,s}=\rank \phi^{t,s} = \dim \ran \phi^{t,s}$ for $s\le t$, or, equivalently, by the numbers $N_{d,b}$, as $r^{t,s} = \sum_{b<s\le t \le d}N_{d,b}$, and $N_{t,s}=r^{t,s+1}-r^{t+1,s+1} - (r^{t,s}-r^{t,s}-r^{t+1,s}) $.

Our persistence module $H(r)=H(V(r))$ is parametrized over $r\in[0,\infty)$, and not over $c\in\ZZ$; however, since we only consider finite metric spaces $X$, we can parametrize over the finite set $\{d(x,y):\,x,y\in X \}$ after sorting it. In an abuse of notation, we will do so implicitly.

In the following we will discuss the visualization of persistence modules, what we mean by approximation, and the need of approximation due to computational scaling.

\subsection{Persistence diagrams}

There are two popular visualizations of persistence modules: First, in \emph{bar codes}, one collects the intervals $(b,d]$ and draws them, annotated with the dimension $k$, and possibly a representative (in the quotient space $H=\ker \partial / \ran \partial$).
If there are many bars stretching over the same interval, one typically draws only a single bar and annotates it with its multiplicity.
In order to understand $H(V(r))$, one then collects all bars containing $r$; in order to understand $\ran\phi^{t,s}$, one collects all bars whose interval is a superset of $[s,t]$. Note that by construction, a bar $(s,t]$ is not a superset of $[s,t]$.

The second popular visualization is as a \emph{persistence diagram}. Instead of drawing bars, i.e.~intervals $(b,d]$, we draw a dot $(b,d)$ in the $\RR^2$-plane. We can read off $\ran\phi^{t,s}$ by collecting all dots $(b,d)$ lying strictly to the left ($b<s$) and upwards $d\ge t$ of the point $(s,t)$. We call this diagram $P_V$.

We can merge these two visualizations by drawing $(b,d]\times \{d\}$ for each basis vector; that is, for each dot in the persistence diagram, we draw a straight line to the right until we hit the diagonal; or, alternatively, we stack the bars, drawn horizontally, using the death $d$ as $y$-coordinate.

\begin{example}[The unit circle]\label{example:s1}
	As a simple example, consider the unit circle $\RR/\ZZ$ with the standard distance and $X_{32}\subset \RR/\ZZ$ as $32$ equispaced sample points. We can describe the persistent homology in terms of bar-codes or as a persistence diagram; for this graphic, we merged both representations, and computed persistent homology up to dimension $10$, c.f.~Figure \ref{fig:s1_32}. From this, we can read off that, e.g.~for scales $r\in(1/32,11/32)$, all embeddings induce isomorphisms on homology, which is given by one dimension in $H_0$ (single connected component), and one dimension in $H_1$ (we roughly have an $S^1$).
	
	It is possible to define a notion of persistent homology of Rips complexes for infinite metric spaces, like $\RR/\ZZ$, and compute such persistence diagrams. We refer to \cite{adamaszek2017vietoris} for the details; we plotted the persistence diagram of the actual sphere $S^1=\RR/\ZZ$ in Figure \ref{fig:s1_real}.

	\begin{figure}[hbpt]
		\centering
		\begin{subfigure}[b]{0.55\textwidth}
			\centering
			\includegraphics[width=\textwidth]{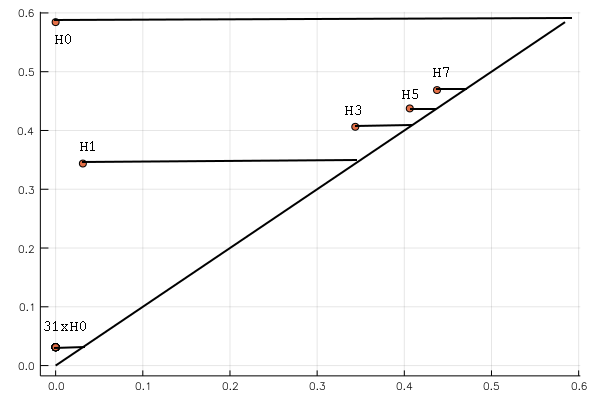}
			\caption{Persistence diagram of $32$ points on $S^1$}\label{fig:s1_32}
		\end{subfigure}%
		\begin{subfigure}[b]{0.55\textwidth}
			\centering
			\includegraphics[width=\textwidth]{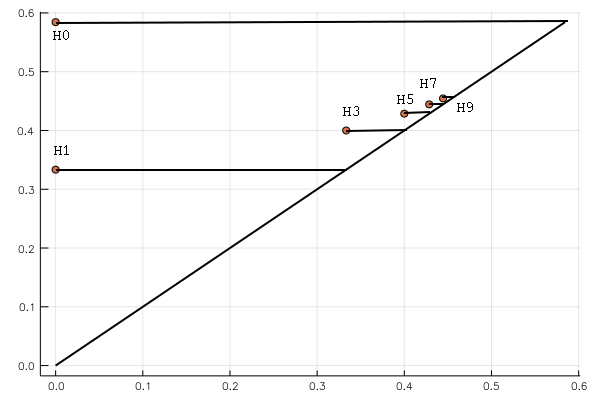}
			\caption{Persistence diagram of $S^1$}\label{fig:clockDiagram}
		\end{subfigure}
		\caption{Persistent homology of the unit circle, up to dimension $10$}\label{fig:s1_real}
	\end{figure}
\end{example}

\subsection{Interleavings and Gromov Hausdorff distance}\label{sec:interleaving-gromov}
The notion of approximation of persistence modules is that of interleavings:

\begin{definition}\label{def:interleave}
	Suppose we have given two persistence modules of $\mathbb F$-vector spaces $V=V(r)$ and $W=W(r)$, and a nondecreasing function $\psi:\RR\to\RR$, with $\psi(r)\ge r$.
	A \emph{$\psi$-interleaving} is a collection of linear maps $\phi_{W,V}^{t,s}$ for $s\le t$ and $\phi_{V,W}^{t,s}$ for $\psi(s)\le t$, such that everything commutes (i.e.~also with $\phi_{V,V}^{t,s}$ and $\phi_{W,W}^{t,s}$ for $s\le t$).
	\begin{equation}\label{diag:interleave}\begin{gathered}
	\xymatrix{
		V(r) \ar[r]^{\phi_{V,V}^{s,r}} &  V(s) \ar[rd]^{\phi_{W,V}^{\psi(s),s}} &\\
		W(r) \ar[u]^{\phi_{V,W}^{r,r}} \ar[rr]^{\phi^{\psi(s),r}_{W,W}}  &&  W(\psi(s))
	}\qquad
	\xymatrix{
		V(r) \ar[dr]_{\phi_{W,V}^{\psi(r),r}} \ar[rr]^{\phi_{V,V}^{s,r}} &&  V(s)\\
		& W(\psi(r)) \ar[r]^{\phi_{W,W}^{s,\psi(r)}} &  W(s) \ar[u]_{\phi_{V,W}^{s,s}}
	}  \end{gathered}
	\end{equation}
	We then say that $W$ is $\psi$-interleaved into $V$ (equivalently: $V$ is $\psi^{-1}$ interleaved into $W$).
\end{definition}

A way of obtaining an interleaving for Rips complexes is by perturbing the distance function $d$: If we have some $\ell:X^2\to \RR_+$ such that $d(x,y)\le \ell(x,y)\le \psi(d(x,y))$ for all $x,y$ then $V_\ell(r)\subseteq V_d(r)\subseteq V_\ell(\psi(r))$ for all $r$. Applying the homology functor turns this interleaving of simplicial flag complexes with embeddings into an interleaving of vector spaces with linear maps.

As a more interesting example, consider a finite metric space $X$, and an $\epsilon/2$-dense  subset $X_\epsilon\subseteq X$, i.e., we have a projection $\pi:X\to X_{\epsilon}$ with $d(x,\pi x)\le \epsilon/2$ for all $x\in X$. Let $V$ denote the Rips-complex of $X$, and $V_{\epsilon}$ denote the Rips-complex of $X_{\epsilon}$. We clearly have $V_\epsilon(r)\subseteq V(r)$ for all $r$; however, we cannot have $V(r)\subseteq V_{\epsilon}(R)$, for any $R$, as an inclusion of simplicial complexes; there simply are not enough points in $X_{\epsilon}$ if $X_\epsilon \subsetneq X$.

This can be remedied by ``multiplexing points''; that is, by considering $\overline V_\epsilon(r)=\{M\subseteq X:\, \diam\pi(M)<r\}$. By the triangle inequality, $d(x,y)\ge d(\pi x,\pi y)-d(x,\pi x) - d(y,\pi y)\ge d(\pi x,\pi y) -\epsilon$, yielding $V_\epsilon(r)\subseteq V(r)\subseteq \overline V_{\epsilon}(r+\epsilon)$, for $r\ge 0$. 

Intuitively, $V_\epsilon(r)$ and $\overline V_{\epsilon}(r)$ should describe the same space, and hence $V_\epsilon(r)\subseteq V(r)\subseteq \overline V_{\epsilon}(r+\epsilon)\simeq V_{\epsilon}(r+\epsilon)$, which gives an interleaving $V_\epsilon(r)\subsetsim V(r)\subsetsim V_{\epsilon}(r+\epsilon)$. The topological basis for this is the following:
\begin{definition}
	Let $K\subseteq L$ be flag-complexes, and let $\pi:L_0\to K_0$ be a projection of the underlying point set. We say that this pair is a \emph{single step deformation} if $\pi$ is \emph{simplicial}, i.e.,~$\{\pi(x),\pi(y)\}\in K$ whenever $\{x,y\}\in L$ (this means that $\pi:L\to K$ is well-defined) and \emph{contiguous to the identity}, i.e.,~$\{x,\pi(x)\}\in K$ for all $\{x\}\in L$, i.e.,~$x\in L_0$.
	
	One says that $K$ and $L$ are \emph{simple homotopy equivalent}, if there exists a sequence of single-step deformations connecting $K$ and $L$.  If $K\subseteq L$, by a sequence of increasing single step deformations, then we write $K\subsetsim L$. 
\end{definition}
Simple homotopy equivalent complexes are homotopy equivalent, when considered as topological spaces of formal convex combinations: This is because single step deformations are homotopy equivalences, by $\pi\circ\iota=\mathrm{id}_K$ and $h$ with $h(0)=\mathrm{id}_L$ and $h(1)=\pi=\iota\circ\pi$ defined by
\[h(t)\left[\sum_i \lambda_i x_i\right] = (1-t)\left[\sum_i \lambda_i x_i\right] + t \left[\sum_i \lambda_i \pi(x_i)\right].\]

Considering the definition of $\overline V$, we see that $V_\epsilon(r)\subsetsim\overline V_{\epsilon}(r)$; and, since homology is invariant under homotopy equivalence, we really obtain an interleaving in homology. 

Using $\epsilon/2$-dense subsets, we naturally obtain an interleaving for 
two finite metric spaces $X$ and $Y$ that have Gromov-Hausdorff distance bounded by $\varepsilon_1+\varepsilon_2$. That is, we assume that we can extend the metric $d$ from $X^2\cup Y^2$ to $(X\times Y)^2$, preserving the triangle inequality, such that $X\subseteq X\cup Y$ is $\varepsilon_1/2$ dense and $Y\subseteq X\cup Y$ is $\varepsilon_2/2$-dense.
Then we obtain an interleaving $V_X(r)\subseteq V_{X\cup Y}(r) \subsetsim V_Y(r+\varepsilon_2)\subseteq V_{X\cup Y}(r+\varepsilon_2)\subsetsim V_{X}(r+\varepsilon_1 +\varepsilon_2)$.

The interleaving is even slightly tighter than this by the following observation:
If $Y$ is $\psi$-interleaved into $X$, modulo simple homotopy equivalence, and $R=\mathrm{diam}\,X = \max\{d(x,x'):\, x,x'\in X\}$, then $Y$ is also $\psi_{R}=\max(R, \psi(r))$-interleaved into $X$, as $V_X(R)$ is contractible.

Consider the example \ref{example:s1} of $X_{32}\subset S^1=\RR/\ZZ$. The set $X_{32}$ is $\frac{1}{64}$-dense. We therefore obtain a $\psi(r)=\max(1/2, r+1/32)$-interleaving, from any finite $X_{32}\subseteq X\subseteq S^1$ into $X_{32}$, modulo simple homotopy. In this sense, the persistence homology of $X_{32}$ approximates, up to $\psi$-interleaving, the persistent homology of $S^1$.

\subsection{Approximate persistence diagrams}\label{sec:approx-pers} Now, suppose that $W$ is $\psi$-interleaved into $V$. Consider the corresponding commutative diagrams \eqref{diag:interleave}.
Given some feature $\zeta\in \ran\phi_{W,W}^{\psi(s),r}$ persisting in $W$ from $r$ to $\psi(s)$, with $r\le s$, this feature must also persist in $V$ from $r$ to $s$: Consider the left diagram in \eqref{diag:interleave}. Take a preimage $\tilde \zeta\in W(r)$, and consider $\tilde\xi=\phi^{r,r}_{V,W}\tilde\zeta$:
\[0\neq\zeta=\phi^{\psi(s), r}_{W,W}\tilde \zeta = 
\phi^{\psi(s),s}_{W,V}\phi^{s,r}_{V,V}\phi^{r,r}_{V,W}\tilde \zeta = \phi^{\psi(s),s}_{W,V}\phi^{s,r}_{V,V}\tilde\xi.\]

Likewise, a feature $\xi\in\ran\phi^{s,r}_{V,V}$ persisting in $V$ from $r$ to $s$ with $\psi(r)\le s$ must persist in $W$ from $\psi(r)$ until $s$.
In terms of ranks, we must have $\rank \phi_{W,W}^{\psi(s),r}\le \rank \phi_{V,V}^{s,r}$ and $\rank \phi_{V,V}^{s,r}\le \rank \phi_{W,W}^{s,\psi(r)}$.

In other words: The set of basis features (from the normal form, i.e.,~the $\xi_{r,i}^{d,b}$) persisting in $W$ from $r$ until $\psi(s)$ form a linearly independent subset of the space of features persisting in $V$ from $r$ until $s$. If $\psi(r)\le s$, then the set of basis features persisting in $W$ from $\psi(r)$ until $s$ span the space of features persisting in $V$ from $r$ until $s$.

In order to visualize this information in the persistence diagram, we draw for each basis-vector $\zeta^{d,b}_{d,i}$ in $W$ a ``bar-with-errors'', i.e. a rectangle $(b,\psi(b)]\times (d,\psi(d)]$, instead of the dot $(b,d)$.
A lower bound on $\rank \phi_{V,V}^{t,s}$ can then be read of by counting all rectangles, that are contained in the upper left corner $[-\infty,t]\times (s,\infty]$.
If $t\le \psi(s)$, we can read off an upper bound by counting all rectangles that intersect the upper left corner $[-\infty,t]\times (s,\infty]$.
For an example, see Figure \ref{fig:24h-pers-approx}.

\begin{figure}[hb]
	\centering
	\includegraphics[width=0.8\textwidth]{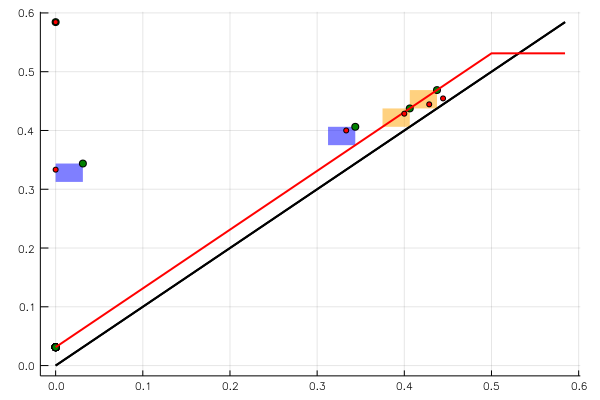}
	\caption{Approximate persistence diagram of 32 points on $S^1$. We plotted the identity in black and $\psi$ in red. We used large green markers for persistence of $X_{32}$, and small red markers for persistence of $\RR/\ZZ$. The shaded rectangles correspond to the error bars of the interleaving. We color-coded the rectangles: Blue rectangles correspond to ``real'' features of $X_{32}$, that must be matched to some feature of $\RR/\ZZ$; orange rectangles may or may not be matched.}\label{fig:24h-pers-approx}
\end{figure}

If we additionally draw the graph of $\psi$, we can partially match rectangles from the persistence diagram of $W$ with dots from the persistence diagram of $V$. This means, we 
\begin{enumerate}
	\item match dots only to rectangles in which they are contained,
	\item the matching may leave some rectangles with upper right corner below the graph of $\psi$ unmatched, but not others,
	\item and the matching may leave some dots below the graph of $\psi$ unmatched, but no others.
\end{enumerate}
This partial matching can be formalized as the \textit{Interleaving Theorem} (see Appendix \ref{app:interleave}). Note that the theorem uses a more 'symmetric' definition of interleavings than Definition~\ref{def:interleave}: Instead of demanding 
$W(r)\subseteq V(r)\subseteq W(\psi(r))$, interleavings are defined with respect to two maps $\psi_{1}$ and $\psi_{2}$ with $W(r)\subseteq V(\psi_1(r))$ and $V(r)\subseteq W(\psi_2(r))$.
This more general case can be again reduced to an interleaving in the sense of Definition~\ref{def:interleave} by defining $W_{\psi_1}(r)\colonequals W(\psi_1^{-1}(r))$ and $\psi\colonequals \psi_1\circ\psi_2$, yielding $W_{\psi_1}(r)\subseteq V(r)\subseteq W_{\psi_1}(\psi(r))$.
For the sake of self-containedness, we give a proof of the interleaving theorem in Appendix \ref{app:interleave}; alternatively see the references in \cite{carlsson2009topology}. The approximate visualization is, to the best of our knowledge, original to this work.

\FloatBarrier

\subsection{Non-sparse computational scaling}
Let us consider how the computation of a full persistence diagram scales with the number $N=|X|$ of sample points in a metric space.
Usually, algorithms will need to have access to the full distance matrix, to enumerate simplices, and to compute the Smith normal form of the boundary matrix (which has simplices as rows and columns).

Firstly, the distance matrix will have $N(N-1)/2 \sim N^2$ entries, each of which needs to be computed and potentially stored in memory. This precludes any ``big data'' applications.

Next, we count the number of simplices: We will have $\binom{n}{k+1}\sim N^{k+1}$ many simplices of dimension $k$ for $k\ll N$, and $2^N$ many simplices in total. Any algorithm that needs to enumerate simplices will become infeasible for computing a full exact persistence diagram, once we have more than 40 sample points. 
This can be somewhat alleviated by instituting a dimension cut-off at some small $k$, bringing the number of $k$-simplices down from exponential to a polynomial $\sim N^{k+1}$. Many algorithms, however, will actually need to enumerate the $k+1$-simplices in order to compute $k$th homology. 

As the linear algebra of how to compute homology and persistence of simplicial complexes is not the focus of this work, we just mention that this scales worst-case in the third power in the number of simplices, in practice however often only linearly.
All homology computations in this paper were done by \texttt{ripser} \cite{bauer2017ripser}, which is a state of the art software for computing persistent (co-)homology over fields $\ZZ_p$ for small primes $p$, given a distance matrix (which does not need to fulfill the triangle inequality).

From these considerations, we can see why the exact computation of persistence diagrams is limited to relatively small data-sets and dimensions. For a survey of current software, see \cite{otter2017roadmap}.

\subsection{Computational scaling of sparsified complexes}\label{sec:sparsi-scale}
As the main bulk of computational scaling is not related to the realm of linear algebra, it is necessary to reduce the number of simplices.
All three sparsification strategies mentioned in this work ---subsampling and thresholding, \cite{sheehy2013linear}, and the strategy introduced in this paper--- exhibit a linear asymptotic scaling of the number of simplices. In pratice we generally observe a speed up as a result of the sparsification, although the asymptotic regime of linear scaling is usually only entered for very small dimensions. This drawback is due to exponential scaling of the constants in the desired precision.

Let us first recall subsampling and thresholding. One fixes a lower bound $\epsilon>0$, and chooses a subset $X_\epsilon\subseteq X$ of sample points with a projection  $\pi:X\to X_\epsilon$, such that $d(x, \pi x)\le \epsilon$ holds for all $x\in X$.
We say that this subsample has \emph{density} $\rho$ if $d(x,y)\ge \rho^{-1}\epsilon$ for all $x\neq y$ in $X_\epsilon$. 
One then thresholds the resulting space by fixing a $T>\epsilon$ and only considering the persistence diagram in the range $r\in [0,T]$.
This means edges $(x,y)$ that have $d(x,y)\ge T$ are taken out of the considerations.
Let $R=\max\{d(x,y): x,y\in X\}$ be the diameter of the space. We have seen in Section \ref{sec:interleaving-gromov} that $V(X_{\epsilon},T)$ is then $\psi$-interleaved into $V(X)$ with
\[\psi(r) = \left\{\begin{aligned}
r+2\epsilon &&\text{if } r\le T-2\epsilon\\
R &&\text{otherwise}.
\end{aligned}\right.\]
It is instructive to rewrite this as a relative error bound: Let $Q=T/\epsilon$ be the quotient of scales; the lowest nontrivial relative error $\epsilon_1(r)=(\psi(r)-r)/r$ is obtained at $r=T-2\epsilon$, for an error of $\epsilon_1^*=2\epsilon/(T-2\epsilon) = 2 Q^{-1} (1-2Q^{-1})^{-1}\approx 2 Q^{-1}$ for relatively large $Q$. 

In order to count simplices, we will need to make some assumptions on the metric space $X$. One standard assumption is the following: 
\begin{assumption}\label{A:intrinsicD}
	The metric space $X$ has bounded intrinsic dimension $\log_2(C)$, i.e., we can place at most $C$ distinct points with pairwise distances bounded below by $r/2$ in the ball $B_r(x)$, for every $x\in X$ and $r>0$.
\end{assumption}
It is clear that all subsets of $\RR^n$ have bounded intrinsic dimension, and that this property is invariant under isometries.

If we assume A~\ref{A:intrinsicD}, we can estimate the number of simplices in $V(X_{\epsilon}, T)$: Each edge must have $\epsilon \rho^{-1}\le d(x,y)\le T$. Assume that $Q=T/\epsilon < \rho^{-1}2^L$; then, each point can have at most $C^L$ neighbors (in the $V(X_{\epsilon}, T)$ considered as a graph).
Therefore, the number of $k$-simplices is bounded above by 
\begin{align*}
N \binom{C^{L}}{k}\leq N C^{Lk}\approx N \left(Q\rho\right)^{k\log_2C}\approx N \left(\frac{2\rho} {\epsilon_1^{*}}\right)^{k\log_2C}.
\end{align*}
From this simple count, we can infer linear scaling in the number of points, albeit polynomially bad scaling in the relative precision, and
exponentially bad scaling in the dimension of simplices and the intrinsic dimension of the dataset.

Note that in subsampling and threshholding the relative precision $\epsilon^{\star}_{1}$ is only obtained close to the threshhold at $r=T-2\epsilon$. Our own approach is to remedy this shortcoming while retaining a similar number of simplices.

For this we prescribe an arbitrary interleaving $\psi$. Usually, this will have the form 
\[\psi(r)=\min(R, r+\max(\epsilon_0, \epsilon_1 r))\]
with absolute precision $\epsilon_{0}$ and relative precision $\epsilon_{1}$.
We can relate $\epsilon_{1}$ to a precision parameter $Q>1$ via the formula $Q=2+2\epsilon_1^{-1}$, or $\epsilon_1=2Q^{-1}(1-2Q^{-1})^{-1}$. We then construct a family of approximating subsets $X_t$, with $X_{t}\subseteq X_{t'}$ for $t'>t$, and bounded density $\rho$, and consider only (a subset of) edges with $d(x,y)\le Q\min(t(x), t(y))$, where $t(x)\colonequals\inf\{t: x\in X_t\}$. In order to bound the number of simplices, we attribute each simplex to one of its vertices with minimal $t$.
Assuming $Q <\rho^{-1} 2^L$, the number of simplices attributed to each point $x$ is bounded above by $C^{Lk}\approx \left(Q\rho\right)^{k\log_2C} $,
because each relevant neighbor $y$ needs to have $t(y)\ge t(x)$ and hence $\rho^{-1}t(x)\le d(x,y)\le Qt(x)$.

The price we have to pay for this global interleaving is somewhat less flexibility in the choice of subsamples: It is trivial to find a $\rho=1$ subsample for fixed $t$ using $|X_t|\,|X|$ evaluations of the metric; we are, however, constrained to \emph{contraction trees}, that carry some additional structure. This yields $\rho=4$, in $|X|\log |X|$ metric evaluations.

\section{Sparse Approximations}
Fundamental to the approximation is the organization of the finite metric space $X$ as a rooted tree, such that the intricate details are at the lower levels, while the higher levels provide a coarse overview.
The structure we use specifically is that of a $Q$-contraction tree.
\begin{definition}
	An \emph{ordered rooted tree} on a finite set is an enumeration, i.e.~an ordering, $X=\{x_0,x_1,\ldots\}$ and a \emph{parent map}, which is a strictly decreasing function $\parent: \{1,\ldots,|X|\}\to \{0,\ldots,|X-1|\}$. In an abuse of notation, we also consider it as a map $\parent:X\setminus\{x_0\}\to X$, via $\parent x_n \colonequals x_{\parent n}$.
	
	We call $X_n=\{x_0,\ldots,x_n\}$ the truncation of the tree.
	The projection $\pi_n:X\to X_n$ onto the truncation is given by going up in the tree until we are in $X_n$, i.e.~by $\pi_n x=\uparrow^k \hspace{-0.3em}x$ with $k=\min\{m: \uparrow^m \hspace{-0.3em}x\in X_n\}$. 
	
	A \emph{contraction tree} is an ordered rooted tree additionally endowed with nonincreasing contraction times $t=t(x)\in[0,\infty]$ with $\infty=t_{0}>t_{1}\ge t_{2}\ge \ldots\ge 0$ for $t_{k}=t(x_{k})$. This defines the function $n(t)=\max\{k: t(x_k)\geq t\}$, and associated truncations $X_{n(t)}$ and $\pi_{n(t)}$.
	
	A contraction tree on a metric space must additionally fulfill the inequality $d(x,\pi_{n(t)}x)\le t$, for all $x\in X$ and $t\ge 0$. If we have the stronger estimate $d(x,\pi_{n(t)}x)\le Q^{-1}(t)$, for some nondecreasing function $Q:[0,\infty]\to[0,\infty]$ with $Q(t)\ge t$, then we call it an \emph{$Q$-contraction tree}, and $Q$ the metric precision function. 
\end{definition}

Given a contraction tree on $X$ and a symmetric weight function $\ell: X^2\to [0,\infty]$, let $V_\ell(n,r)=\left\{M\subseteq X_n:\,\ell(x,y)< r\,\,\,\forall x,y\in M\right\}$ be the restriction of the Rips-complex $V_\ell(r)$ to $X_n$. Hence, we have $V_\ell(|X|, r)=V_\ell(r)$, and $V_\ell(n-1,r)\subseteq V_\ell(n,r)$ and $V_\ell(n,r)\subseteq V_\ell(n,r')$ for $r\le r'$.

The construction of sparsified complexes will be in reverse order:
First, we construct a sparsified complex from a given $Q$-contraction tree. Then we discuss how to prescribe the metric precision function $Q$ in order to obtain desired relative and absolute errors, and how to get this $Q$-contraction tree from an $\mathrm{Id}$-contraction tree. Finally, we construct the $\mathrm{Id}$-contraction tree from a combinatorial tree, and this combinatorial tree from the metric space.

\subsection{The sparsified complex}
Consider the parameter $Q$ fixed. Our goal is to use an $Q$-contraction tree on a finite metric space $(X,d)$ in order to obtain a sparse approximation $\widetilde V(r)$ of the Rips complex $V_d(r)=\{(x,y)\in X^2: d(x,y)< r\}$, which is interleaved modulo simple homotopy equivalence.

The general idea is to construct a sparsified complex by definition of a length function $\ell(x,y)\in \{\infty, d(x,y)\}$, such that most edges will be removed from the complex by having length $\infty$, while the others keep their $d$-value. This 
construction will ensure that $V_\ell(|X|, r)\subsetsim V_\ell(n(r), r)$, and corresponds to the subsampled complex in Section \ref{sec:interleaving-gromov}. Next, we need an analogue of the multiplexed complex $\overline V$; this will be $V_{\overline \ell}$, which fulfills $V_{\overline\ell}(|X|, r)\subsetsim V_{\overline\ell}(n(r), r)=V_\ell(n(r),r)\supsetsim V_\ell(|X|, r)$, and otherwise attempts to minimize the length function $\overline\ell$. The crucial estimate to obtain the interleaving
\[
V_\ell(r)\subseteq V_d(r) \subseteq V_{\overline\ell}(\psi_Q(r))\simeq V_\ell(\psi_Q(r))
\]
is $\overline\ell(x,y)\le \psi_Q(d(x,y))$, which will be shown in Lemma \ref{lemma:interleave-Qpsi}. 

\begin{definition}[Sparsified and Implied Complexes]\label{def:sparse-complex}
	The sparsified complex $V_\ell$ is characterized by the subset of missing edges where $\ell(x,y)=\infty$; for all other edges, we have $\ell(x,y)=d(x,y)$. The implied complex $V_{\overline\ell}$ is defined by the implied length $\overline\ell$; it has $\ell,\overline\ell(x,y)=\ell(x,y)=d(x,y)$ for all non-missing edges. The two complexes are constructed recursively (by a depth-first dual tree traversal, see Section \ref{sec:impl-notes}): In order to assign $\ell,\overline\ell(x_i,x_j)$ with $i<j$, we need to know $\ell,\overline\ell(x_i,\parent x_j)$. We distinguish the following four cases:
	\begin{enumerate}[(a)]
		\item If $\ell(x_i,\parent x_j)=\infty$, then $\ell(x_i,x_j)=\infty$, and $\overline\ell(x_i,x_j)=\overline\ell(x_i,\parent x_j)$.
		\item If $\ell(x_i,\parent x_j)<\infty$, and $t(x_j)\leq d(x_i,\parent x_j)$, then  $\ell(x_i,x_j)=\infty$ and $\overline\ell(x_i,x_j)=d(x_i,\parent x_j)$.
		\item If $\ell(x_i,\parent x_j)<\infty$, and $d(x_i,\parent x_j)< t(x_j)< d(x_i,x_j)$, then $\ell(x_i,x_j)=\infty$ and $\overline\ell(x_i,x_j)=t(x_j)$.
		\item If none of the above apply, i.e.~if $\ell(x_i,\parent x_j)<\infty$ and $t(x_j)\geq\max(d(x_i,x_j), d(x_i,\parent x_j))$, then $\ell(x_i,x_j)=\overline\ell(x_i,x_j)=d(x_i,x_j)$.
	\end{enumerate}
	For the sake of this definition, we consider $\ell(x,x)=0$ to be always non-missing.
\end{definition}
By construction, $d\le \ell$ and $\overline\ell\le \ell$. Sparsity of $\ell$ comes from $t(x_j)\geq d(x_i,x_j)$ for all non-missing edges. By the condition $d(x,\parent x)\le t(x)$ for all $x\in X\setminus\{x_0\}$, it follows that edges $(x,\parent x)$ are always non-missing, and we can recursively enumerate all non-missing edges by starting from $(x_0,x_0)$ by traversing a tree.

\begin{lemma}
	The sparsified and implied complexes have the following properties:
	\begin{enumerate}
		\item Suppose $k>n(r)$. Then the inclusion $V_\ell(k-1,r)\subseteq V_\ell(k,r)$ is a single step deformation with $\pi(x,y)=\pi_{k-1}(x,y)\colonequals(\pi_{k-1}x,\pi_{k-1}y)$.
		\item Suppose $k>n(r)$. Then the inclusion $V_{\overline\ell}(k-1,r)\subseteq V_{\overline\ell}(k,r)$ is a single step deformation with $\pi(x,y)=\pi_{k-1}(x,y)\colonequals(\pi_{k-1}x,\pi_{k-1}y)$.
		\item The complexes $V_\ell(n(r), r)=V_{\overline \ell}(n(r), r)$ coincide.
	\end{enumerate}
\end{lemma}
\begin{proof}
	Assume without loss of generality that $i<k$. 
	\begin{enumerate}
		\item The two complexes $V_{\ell}(k-1, r)\subseteq V_\ell(k,r)$ differ only by the single point $x_k$. We have already seen that $\ell(x_k,\parent x_k)=d(x_k,\parent x_k)\leq t(x_k)< r$; hence, the induced $\pi_{k-1}: V_\ell(k,r)\to V_{\ell}(k-1,r)$ is contiguous to the identity. In order to see that it is simplicial, let $i<k$ be such that $\ell(x_k,x_i)<r <\infty$. Then by construction (case (d)), $\ell(\parent x_k, x_i)=d(\parent x_k, x_i)\leq t(x_k)< r$.
		\item As in (1) the two complexes differ only by the single point $x_{k}$. Moreover, from the definition of $\ell$ and $\overline{\ell}$, we see that they differ only in the missing edges. Thus, contiguity to the identity follows as before, and for simpliciality we only have to consider the case $\overline \ell(x_k,x_i)<r <\ell(x_k,x_i)=\infty$. This rules out (d) and (a).
		If the case $(b)$ applies, then $\overline\ell(x_{i},\parent x_{k})\leq \ell (x_{i},\parent x_{k})=d(x_{i},\parent x_{k})=\overline\ell(x_{i}, x_{k})<r$.
		If the case $(c)$ applies, then $\overline\ell(x_{i},\parent x_{k})\leq \ell (x_{i},\parent x_{k})=d(x_{i},\parent x_{k})<t(x_{k})<r$.
		\item As $\overline\ell\leq \ell$, we have trivially $V_\ell(n(r), r)\subseteq V_{\overline \ell}(n(r), r)$.
		Thus, we have to show that $\overline\ell(x_i,x_j)<r$ implies $\ell(x_i,x_j)<r$ for all $x_i,x_j\in X_{n(r)}$, i.e. for all $x_i,x_j$ with $t_{i},t_{j}\geq r$. 
		This is equivalent to showing that the edge $(x_i,x_j)$ is non-missing
		if $ \overline\ell(x_{i},x_{j})<r\leq t_{j}\leq t_{i}$.
		Now, if $\ell(x_i,x_j)=\infty$ for $i<j$, then either (a), (b), or (c) apply. If either (b) or (c) applies, it follows directly that $\overline\ell(x_{i},x_{j})\geq t_j$. If (a) applies, then $\overline\ell(x_{i},x_{j})\geq t_j$ as well by induction and monotonicty of the $t_j$. Thus, the claim follows.
	\end{enumerate}
\end{proof}
The critical estimate to complete the interleaving is the following:
\begin{lemma}\label{lemma:interleave-Qpsi}
	Suppose that we have a $Q$-contraction tree, i.e., $d(x,\pi_{n(r)}x)\le Q^{-1}(r)$ for all $r\ge 0$ and $x\in X$.
	
	Then, for all $(z,y)\in X^2$, we have $d(z,y)\ge \overline\ell(z,y)-2Q^{-1}(\overline\ell(z,y))$. Therefore $V_{d}(r)\subseteq V_{\overline\ell}(\psi_Q(r))$, where $\psi_Q^{-1}= 1-2Q^{-1}$.
\end{lemma}
\begin{proof}
	Let without loss of generality $(x_i,x_j)\in X^{2}$ with $i<j$. We will consider the four cases in the definition of $\overline\ell$ in reverse order:
	\begin{enumerate}
		\item[(d)] As $\ell(x_i,x_j)=d(x_i,x_j)=\overline\ell(x_i,x_j)$, the claim is trivial.
		\item[(c)] In this case we have $d(x_i,x_j)>t_j=\overline\ell(x_i,x_j)$.
		\item[(b)] In this case we have $d(x_i,x_j)> d(x_i,\parent x_j)=\overline\ell(x_i,x_j)$.
		\item[(a)] This case is handled recursively: 
		For $x_{i},x_{j}\in X$ with $\ell(x_i,\parent x_j)=\infty$, find  $x_{k},x_{n}\in X$ such that $n>k$, $\pi_{n}x_j=x_n$ and $\pi_{n}x_i = x_k$, $\ell(x_k, x_n)=\infty$, and $\ell(x_k, \parent x_n)<\infty$.
		Then $\overline\ell(x_{i},x_{j})=\overline\ell(x_{k},x_{n})$
		and  
		\begin{enumerate}
			\item[(c)] either $\overline\ell(x_{k},x_{n})=t_{n}$, and
			\begin{align*}
			d(x_i,x_j)&\geq d(x_n,x_k) - d(x_i, x_k) - d(x_j,x_n)
			>t_n - 2Q^{-1}t_{n}\\
			&= \overline\ell(x_{i},x_{j}) - 2Q^{-1} \overline\ell(x_{i},x_{j}),
			\end{align*}
			\item[(b)] or $\overline\ell(x_{k},x_{n})=d(x_k,\parent x_n)\leq t_{n-1}$ (as $(x_k,\parent x_n)$ is non-missing), and
			\begin{align*}
			d(x_i,x_j)&\geq d(x_k,\parent x_n) - d(x_k,x_i) - d(x_j,\parent x_n)\\
			&\geq \overline\ell(x_{i},x_{j}) - Q^{-1}t_{n}- Q^{-1}t_{n-1}\\
			&\geq \overline\ell(x_{i},x_{j}) - 2Q^{-1} \overline\ell(x_{i},x_{j}).&&\qedhere
			\end{align*}
		\end{enumerate}
	\end{enumerate}
\end{proof}

\subsection{Relative error interleavings}\label{sec:relerrorInterleave}
Assume that the dataset $X$ is organized into a contraction tree with contraction times $\infty=r_0>r_1 \ge r_2 \ge \ldots$ and $d(x, \pi_n x)\le r_n$ for all $x\in X$. 
Lemma~\ref{lemma:interleave-Qpsi} permits us, in principle, almost arbitrary choices of the interleaving-$\psi$ and the metric approximation $Q$, as functions. A very natural choice are interleavings with prescribed relative errors, i.e.~of the form $\psi(r)=(1+\epsilon_1)r$. In order to obtain such an interleaving, we construct a $Q$-contraction tree by setting contraction times $t_n = Q(r_n)$, where $Q(r)=(2+2\epsilon_1^{-1})r$. In view of Lemma~\ref{lemma:interleave-Qpsi}, this yields the estimate $d(x,y)\ge (1 - 2(2+2\epsilon_1^{-1})^{-1} )\overline \ell(x,y) = (1+\epsilon_1)^{-1}\overline\ell(x,y)$, and hence $\overline\ell\le (1+\epsilon_1)d$, which is the desired estimate. 

As before, the interleaving error actually cuts off at the radius: For $r>R\colonequals r_1\ge \max\{d(x_0,x):\, x\in X\}$, the complex $V_d(r)$ is contractible. This is also the case for $V_\ell$: We know that $\overline\ell(x_0,x)\le R$ for all $x$, by considering cases $(a)-(d)$. Hence, we actually have the better interleaving estimate $\psi(r)=\max(R, (1+\epsilon_1)r)$.

If data is plentiful, it is often necessary to further restrict the number of samples: That is, we truncate at some $N<|X|$, in addition to the multiplicative error $\epsilon_1$. This means, we set $t_n=(2+2\epsilon_1^{-1})r_n$ for $n\le N$, and $t_n=r_n$ for $n\ge N$, and hence construct a $Q$-contraction-tree with 
\[Q^{-1}(r) = \left\{\begin{array}{ll}
(2+2\epsilon_1^{-1})^{-1}r,\quad & \text{if } (2+2\epsilon_1^{-1})r_{N+1}\le r\\
r_{N+1},\quad & \text{if } r_{N+1}\le r \le (2+2\epsilon_1^{-1}) r_{N+1}\\
r,\quad & \text{if } r\le r_{N+1} .
\end{array}\right.
\]
Now, applying Lemma~\ref{lemma:interleave-Qpsi} with a cut-off at $d\ge 0$, using the shorthand $\epsilon_0\colonequals 2r_{N+1}$, we obtain the estimate $d(x,y)\ge \psi^{-1}(\overline\ell(x,y))$, where
\[\psi^{-1}(r) = \left\{\begin{array}{ll}
(1-(1+\epsilon_1^{-1})^{-1})r,\quad& \text{if } (1+\epsilon_1^{-1})\epsilon_0\le t\\
r-\epsilon_0,\quad& \text{if } \epsilon_0\le r \le (1+\epsilon_1^{-1}) \epsilon_0\\
0 ,\quad& \text{if } r\le \epsilon_0,
\end{array}\right.
\]
and hence $V_{d}(r)\subseteq V_{\overline\ell}(\psi(r))$ with 
\[\psi(r) = \left\{\begin{array}{ll}
r+ \epsilon_1r,\quad& \text{if } \epsilon_0\epsilon_1^{-1}\le r\\
r+ \epsilon_0,\quad& \text{if } r \le \epsilon_0\epsilon_1^{-1}.
\end{array}\right.
\]  
In other words, and combining with the previous discussion, we obtain the estimate with
\begin{equation}\label{eq:interleave-rel-abs-err}
\psi(r)= \min(R, r + \max(\epsilon_0 , \epsilon_1 r)),\quad \epsilon_0=2r_{N+1}.
\end{equation}
We can then consider the truncated $V_\ell(N, r)$ only, and obtain the above interleaving.

If we aim for a different desired precision $\psi$, then we can obtain the required $Q^{-1}$ from Lemma \ref{lemma:interleave-Qpsi}. Analogous to the above, an $\mathrm{id}$-contraction tree with times $\{r_n\}$ can be turned into a $Q$-contraction-tree with times $\{ t_n \}$ by setting $t_n=Q(r_n)$, i.e. by relabeling.

\subsection{Contraction trees via simplified cover-trees}
We have shown how a contraction tree (with $Q=\mathrm{id}$) on a finite metric space $X$ yields a sparsified approximate Rips-complex that is interleaved with arbitrary desired precision $\psi$. In Section \ref{sec:sparsi-scale}, we argued upper bounds for the number of simplices in the sparsified complex for desired precisions of the form \eqref{eq:interleave-rel-abs-err}. These estimates required that the metric space has bounded density. For contraction trees we define this property as follows.
\begin{definition}
	Let $X$ be a contraction tree on a finite metric space. We say that $X$ has \emph{density} bounded above by $\rho=\rho(X)$, if $d(x_i,x_j)\ge \rho^{-1} t_j$ for all $i<j$.
\end{definition}
For simplicity, we only consider constant $\rho$ instead of general functions $\rho(r)$.

In order to obtain the desired sparsification, we need to construct contraction trees with bounded density. The well-known cover-tree \cite{beygelzimer2006cover} is such a tree with $\rho=4$, and can be directly used. 

One can obtain a contraction tree with $\rho\le 4$ also by the following simplified variant of the cover-tree, which we shall introduce for the sake of self-containedness and ease of implementation. The simplified cover-tree fulfills $d(x,\parent x)\le \frac{1}{2}d(\parent x, \grandparent x)$ for all points. By the geometric series, we obtain a contraction tree by setting $r(x)=2d(x,\parent x)$, using the convention that $d(x_0,\parent x_0)=\infty$.

We can ensure a density of $\rho\le 4$, i.e.~$d(x,y)\ge \min(r(x),r(y))/4$ for all $x\neq y$, by constructing the tree sequentially, in the following way: When inserting a new point $x$ into a tree $X$, we choose its parent \[\parent x = p\colonequals \argmin_{y\in X}\{d(x,y): d(x,y)\le d(y,\parent  y)/2\},\]
and set $r(y)\colonequals 2d(y,p)$. For non-empty $X$, the right hand side always contains $x_0$ and is therefore non-empty. Now, suppose that the density bound was violated after insertion of $x$; then there exists a $y\in X$, such that $2 d(x,y)\le \min(d(x,p), d(y,\parent y))$. Now $2d(x,y)\le d(y,\parent y)$ means that $y$ is a valid candidate for the minimization, and hence $d(x,y)\ge d(x,p)$; this contradicts the assumption, and by induction we always construct a tree with density $\rho\le 4$.

In order to insert a point $x$ into $X$, we need to search the current tree. The search-tree can be aggressively pruned, because $x$ can never be inserted below $y$, i.e., we cannot have $x\preceq y$, if $d(x,\parent y)> r(y)$. In other words, we obtain the following algorithm for finding the minimizing parent of a new point $x$:
\begin{enumerate}
	\item Initialize current optimal solution as $(x^*,d^*)=(\texttt{null}, \infty)$. Initialize the candidate set as $C=\{x_0\}$.
	\item If $C=\emptyset$, then return $(x^*,d^*)$. Otherwise, take some point $y\in C$. This is the current candidate. Compute $d=d(x,y)$. If $d\le r(y)$ and $d<d^*$, then update $(x^*,d^*)= (y,d)$.
	\item Consider all children $(c,r(c))$ of $y$, in descending order of $r(c)$. If $d\le r(c)$, then add $c$ to the candidate set $C$; otherwise, or if all children have been processed, go to step $(2)$.
\end{enumerate}

This algorithm needs asymptotic $N\log N$ distance computations in order to organize a finite metric space $X$ with $N=|X|$ points into a tree, as long as $X$ has bounded doubling dimension. We recommend \cite{beygelzimer2006cover} for a more detailed discussion of the asymptotic run-time. In practice, sequential insertion into simplified cover-trees is fast enough to be negligible compared to the homology computation. Faster construction algorithms exist and are discussed in \cite{beygelzimer2006cover}, but are not in the scope of this paper.

There is a simple but very powerful optimization that is possible on contraction trees: We can replace the a priori bounds $r(x)$ by a posteriori bounds $\orad(x)$. In order to do this, we forget the ordering and the $r_n$, and only keep the combinatorial rooted tree. Then we can naturally construct an $\mathrm{id}$-contraction tree in the following way:
\begin{enumerate}
	\item For each non-root, set $R_0(x)=\max\{d(y,\parent x): y \preceq x\}$, where $y\preceq x$ (``$y$ is descendant of $x$'') if either $y=x$ or if $x$ lies on the unique path between $y$ and $z$ in the tree , i.e.~in the graph with edges $\{(z,\parent z)\}$.
	\item We set $\orad(x)=\max\{R_0(y): y\preceq x\}$ for each $x\neq x_0$, and set $\orad(x_0)=\infty$.
	\item We order the set of points by $\orad(x)$. In cases of ties, we order by demanding that $\parent x_j = x_i$ implies $j>i$; that is, we extend the partial order $\preceq$ in an arbitrary way.
\end{enumerate}
It is clear that the resulting $\orad(x)\le r(x)$, i.e.,~we have tightened the estimates and decreased the density.

\subsection{Implementation notes}\label{sec:impl-notes}
We will now walk through the entire construction, this time with focus on the implementation. 
\begin{enumerate}
	\item We organize our sample set $X$ in a contraction tree. We do this by sequential construction, using $r(x)=2d(x,\parent x)$, as described in the previous section. In order to run this algorithm, each node in the tree needs to store both the associated point in the metric space, and a sorted (descending by $r$) linked list of children.
	\item We compute the a posteriori distances $\orad$, and sort the vertices. If desired, we truncate to some $0<N\le |X|$; then, we need to store $\epsilon_0=2r_{N+1}$ as the incurred error. Note the discontinuity in the error estimate: For $N=|X|$, we effectively have $r_{N+1}=0$, and the truncation to any $k<N$ is expensive, i.e.~incurs significant errors. On the other hand, in many applications $X$ is conceptually a finite approximation of an infinite space $M$. Then, $\epsilon_0$ can also serve as an estimate of the error incurred by considering $X$ instead of $M$. A detailed statistical analysis is beyond the scope of this work, but if $X$ is a finite random sample out of some distribution in a metric space, then we recommend keeping less than $1/2$ of the sample points, at least if $\epsilon_0$ intended to serve as such an estimate and data is plentiful.
	\item We construct the sparsified complex. Suppose the desired precision $\epsilon_1$ is given. We define $Q$ as in Section~\ref{sec:relerrorInterleave}.
	The set of pairs $(x_i,x_j)$ with $i\le j$ forms a tree (called the dual tree), rooted at $(x_0,x_0)$, with the parent map $\parent (x_i,x_j)= \texttt{sort}(x_i,\parent x_j)$. We traverse this tree, checking at each edge whether Definition~\ref{def:sparse-complex} $(b),(c),(d)$ applies (with respect to the relabeling); if $(d)$ applies, we accept the edge and continue the tree traversal. There is no need to ever construct $\overline\ell$; it serves theoretical purposes only. 
\end{enumerate}

The resulting structure is a sparse symmetric matrix for $\ell$, with implicit zeros on the diagonal, implicit $\infty$ on all other missing positions, and $d(x_i,x_j)$ on all present positions. In order to speed up later computations, we can now further truncate this matrix. 

This sparse matrix is the input to a tool for persistent homology computation. We modified the popular and fast \texttt{ripser} tool in order to accept sparse input matrices. We remark that, due to details of the internal representation of simplices, \texttt{ripser} is limited to dimensions $d$ and sample sizes $N=|X|$ such that $\binom{N}{d+2}< 2^{64}$.

Afterwards, the output of \texttt{ripser} needs to be interpreted. Statistical analyses of the output are beyond the scope of this paper; we recommend the visualization described in Section \ref{sec:approx-pers}.

\texttt{ripser} does not support zig-zag filtrations. If one is using software with such support, it might be worthwhile to consider the filtration $V_\ell(n(r),r)$ instead of $V_\ell(0,r)$. 

\section{Computational Example}

\paragraph{The solenoid}
As an example, consider the solenoid (see Figure \ref{fig:solenoid}). We construct the solenoid in the following way: Consider the map $\Phi: \RR/\ZZ \times [-1.5,1.5]^2 \to\RR/\ZZ \times [-1.5,1.5]^2 $, given by $(\phi,x,z)\to \left(2\phi, x/3+\cos(2\pi\phi), z/3+\sin(2\pi\phi)\right)$. We iterate this map, and consider the attractor $\mathcal S=\bigcap_n \ran\Phi^n$. As an alternative view, we take a continuous injective map from the full torus into itself, given by $\Phi: S^1\times [-1.5,1.5]^2\to S^1\times[-1.5,1.5]^2$, which winds around twice, and is contracting in the $(x,z)$-plane. This attractor is called the solenoid, and is a well-known fractal: It is connected, but not path-connected, and every slice $\phi=\mathrm{const}$ is a cantor-set. We can embed the solenoid into $\RR^3$ by rotating it around the $z$-axis, e.g.~by embedding $\psi:(\phi,x,z)\to (\cos(2\pi\phi)(1+x/3), \sin(2\pi\phi)(1+x/3) , z)$. We visualize the construction in Figure \ref{fig:solenoid}, by taking a curve with increasing angle $\phi$, and random $(x,z)$-coordinates. We then apply the $\Phi$-map multiple times, and embed the resulting curve into $\RR^3$ via $\psi$. This yields a better visualization than a pure scatterplot of sample points.

What persistent topology do we expect? At very large scales, the solenoid will be a point, just like all metric spaces of bounded diameter; if we consider a finite sample set, it will be a cloud of isolated points at very small scales. As we zoom out, from very close to far scales, we will mostly see an $S^1$ with a little bit of short-lived noise. But the inclusion maps are actually described by $x\to 2x$ on first homology: A homology class that winds once around $\ran\phi^{n+1}$ corresponds to a class that winds twice around $\ran\phi^n$. That is, the maps of persistent homology over scales should recover the construction of the solenoid. 

What does this mean for persistence diagrams? If we consider persistent homology over $\ZZ_2$ (or any field of characteristic 2), then the inclusions induce the zero-map on the first homology, and we therefore expect a sequence of touching medium-length bars, plus topological noise. If we instead consider the question over $\ZZ_3$ (or any field that does not have characteristic 2), then we expect one very long-lived one-dimensional homology class, plus noise. The zeroth homology is expected to consist of short-linved noise, near the density of our sample set, and one class (one connected component) persisting forever. This intuitive description matches with the computed persistence diagram in $H^1$ of a large sample set, c.f.~Figure \ref{fig:solenoid-pers}.

\begin{figure}[hbpt]
  \centering
         \begin{subfigure}[b]{0.35\textwidth}
                 \centering
                 \includegraphics[width=\textwidth]{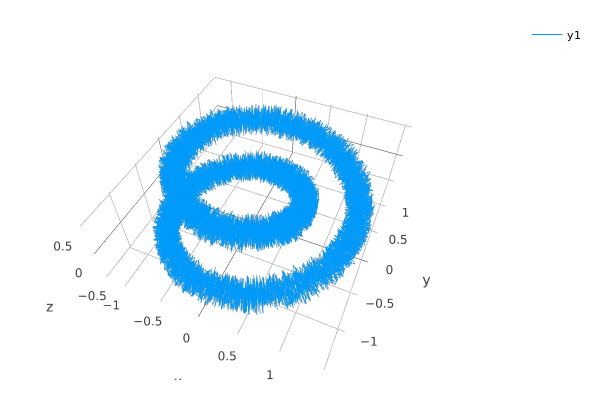}
                 \caption{One iteration}
         \end{subfigure}~~%
         \begin{subfigure}[b]{0.35\textwidth}
          \centering
          \includegraphics[width=\textwidth]{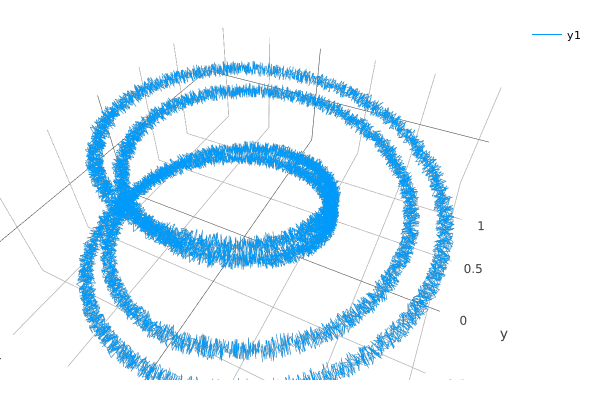}
          \caption{Two iterations}
  \end{subfigure}~~%
  \begin{subfigure}[b]{0.35\textwidth}
    \centering
    \includegraphics[width=\textwidth]{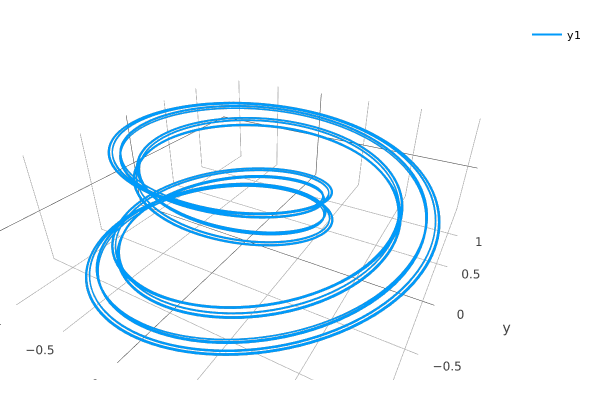}
    \caption{The attractor, many iterations}
\end{subfigure}~~%
         \caption{The solenoid, in several iterations}\label{fig:solenoid}
\end{figure}

\begin{figure}[hbpt]
  \centering
         \begin{subfigure}[b]{0.55\textwidth}
                 \centering
                 \includegraphics[width=\textwidth]{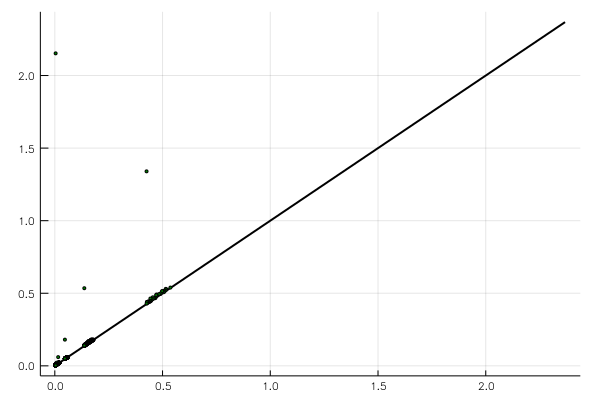}
                 \caption{Diagram over $\ZZ_2$}
         \end{subfigure}~~%
         \begin{subfigure}[b]{0.55\textwidth}
          \centering
          \includegraphics[width=\textwidth]{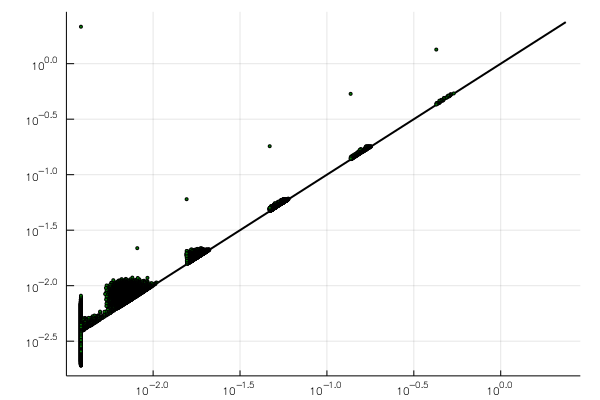}
          \caption{Diagram over $\ZZ_2$, log scale}
        \end{subfigure}\\%
  \begin{subfigure}[b]{0.55\textwidth}
    \centering
    \includegraphics[width=\textwidth]{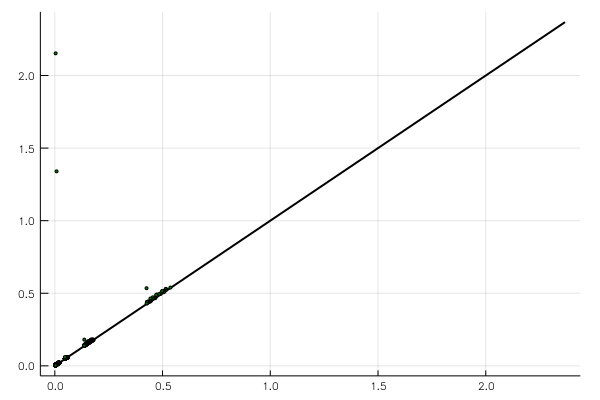}
    \caption{Diagram over $\ZZ_3$}
\end{subfigure}~~%
\begin{subfigure}[b]{0.55\textwidth}
  \centering
  \includegraphics[width=\textwidth]{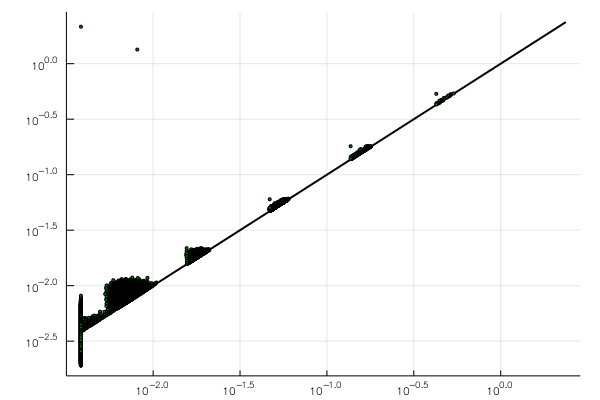}
  \caption{Diagram over $\ZZ_3$, log scale}
\end{subfigure}~~%
         \caption{Persistence diagrams of $H^0$ and $H^1$ of the solenoid (sparsified). Very small scales were truncated for the logarithmic plot; the dots below the diagonal are artifacts}\label{fig:solenoid-pers}
\end{figure}

\paragraph{Approximate persistence diagrams}
From viewing the depicted persistence diagram Figure \ref{fig:solenoid-pers}, it is clear that logarithmic axis scaling is often preferable to absolute axis scaling. However, any practicioner should immediately become doubtful at this diagram: calculations to describe the change of topology over that many orders of magnitude is very hard, if not impossible without a sparsification scheme.

Indeed, Figure \ref{fig:solenoid-pers} does not truly depict the persistence of a very large sample set in the solenoid, but the persistence diagram of an interleaved complex. We took two million random sample points on the solenoid, and then selected a small subsample that is close in Hausdorff-distance $\epsilon_{0}$ to the entire sample set.
This gives a lower bound on the Hausdorff distance to true (uncountably infinite but compact) solenoid, which we also use as a naive estimate of said distance.

We employed two different approximation strategies: The naive strategy (subsampling without threshholding) takes the $1500$ most significant sample points, and obtains an interleaving with $\psi(r)=\min(R,r+2\epsilon_{0})$. With the sparsification strategy developed in this paper, we can use an order of magnitude more, namely $64000$, most significant sample points, and a planned relative error of $\epsilon_1=0.25$, for a resulting interleaving with $\psi(r)=\min(R, r+\max(\epsilon_0, r\epsilon_1))$. Since we have more sample points than in the naive subsample, we can achieve a better $\epsilon_0$, while sacrificing some precision at large $r$. Both computations took comparable amounts of memory, and hence illustrate the trade-off. 

These diagrams also serve to illustrate approximate diagrams. Each entry $(b,d)$ of the persistence diagram of the interleaved complex corresponds to the rectangle $((\psi^{-1}(b),b),(\psi^{-1}(d),d))$, hence sitting at the upper right corner. Such a rectangle definitely corresponds to an actual entry of the true persistence diagram (sitting inside the rectangle), if its lower left corner is below the diagonal (drawn in red), i.e.~if its its upper right corner is above the graph of $\psi$ (drawn in red). These rectangles have been drawn in blue; the other ones are drawn in orange. 

Each entry of the true persistence diagram corresponds to a rectangle (containing the true dot), of either color, if the true entry is below the graph of $\psi$. In other words, the region between the diagonal and $\psi$ is the region of the unknown.

In order to better illustrate the trade-off, we have, in each diagram, also plotted the $\psi$-function of the other approximation. The truncation to the left side of some of the figures is purely a plotting artifact: Both the graphs and the rectangles stretch to $b=0$.

\begin{figure}[hbpt]
  \centering
         \begin{subfigure}[b]{0.5\textwidth}
                 \centering
                 \includegraphics[width=\textwidth]{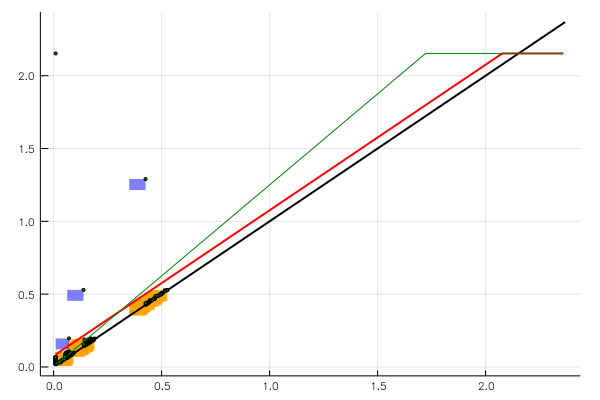}
                 \caption{Subsampled solenoid}
         \end{subfigure}~~%
         \begin{subfigure}[b]{0.5\textwidth}
          \centering
          \includegraphics[width=\textwidth]{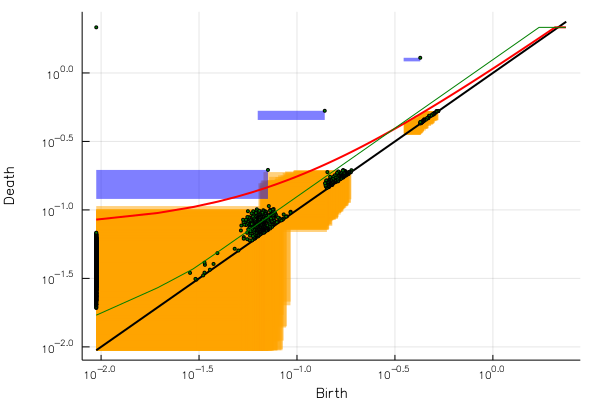}
          \caption{Subsampled solenoid, log scale}
        \end{subfigure}\\%
  \begin{subfigure}[b]{0.5\textwidth}
    \centering
    \includegraphics[width=\textwidth]{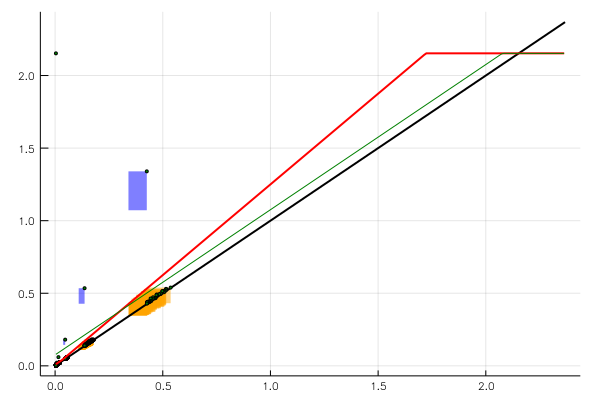}
    \caption{Subsampled and sparsified solenoid}
\end{subfigure}~~%
\begin{subfigure}[b]{0.5\textwidth}
  \centering
  \includegraphics[width=\textwidth]{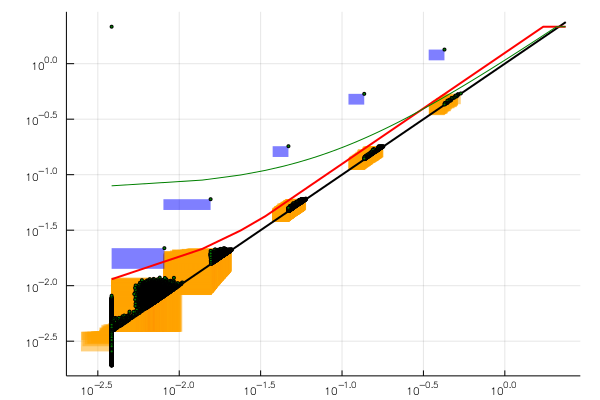}
  \caption{Subsampled and sparsified solenoid, log scale}
\end{subfigure}~~%
         \caption{Persistence diagrams of $H^1(\ZZ_2)$ of the solenoid, with error bars. The identity is in black, the respective $\psi$ is in red and the green line represents the alternative $\psi$, in order to visualize the tradeoff. Rectangles are error bars around features; they are colored blue for definitely present features, and orange for possibly present features.}\label{fig:solenoid-pers2}
\end{figure}

\FloatBarrier

\bibliography{./litlist}
\bibliographystyle{alpha}%

\appendix

\section{Proof of the Classification Theorem for Persistence Modules}\label{app:classify}

We prove Theorem~\ref{T:classification} by construction of a normal form basis $\{\xi^{d,b}_{c,i}\}$:
Let $\{V(c)\}_{c\in\mathbb{Z}}$ be a sequence of finite-dimensional $\mathbb{F}$ vector spaces, such that only finitely many $V(c)\neq \{0\}$, and let $\phi^{c+1,c}: V(c)\to V(c+1)$ be linear maps.

A normal form basis can be obtained by the following procedure, where we keep updating a partial normal form collection $M\subseteq \{\xi^{d,b}_{c,i}\}$ of linearly independent vectors:
\begin{enumerate}
	\item We iterate from $b=-\infty$ to $b=+\infty$.
	\item At index $b$, we iterate from $d=b+1$ until $d=\infty$.
	\item At index $b<d$, we look for a vector $\xi\in \ker \phi^{d, b+1}\subseteq V({b+1})$, which is not part of the span of the current collection $M$. If no such vector exists, we continue with the next $d$.
	\item If we have found such a vector $\xi\in \ker \phi^{d, b+1}$, we add $\{\xi^{d,b}_c = \phi^{c,b+1}\xi: b<c\le d\}$ to our collection.
\end{enumerate}
Properties $(1)$ and $(2)$ hold by construction.
Furthermore, we cannot have created a linear dependency in spaces after $d$. If $\phi^{c,b+1}\xi$ was linearly dependent on $M$, for some $c\le d$, we would find some nonzero $\zeta_{b+1}\in V({b+1})$ in the current span of $M$, such that $\phi^{c,b+1}(\zeta_{b+1}-\xi)=0$.
As we have ensured that $M$ spans $\ker\phi^{c,b+1}$ for all $c\le d+1$, this cannot be. Thus, property~$(1)$ holds.

Note that all the iterations are, in fact, finite, because we assumed that only finitely many $V(c)\neq\{0\}$.
We can run the same procedure for persistence modules $V(r)$ indexed over $r\in\RR$, as long as $V(r)$ changes at only finitely many scales $r$.
\section{The Interleaving Theorem}\label{app:interleave}

Let $V$ and $W$ be two persistence modules, parametrized over $\ZZ$, with only finitely many non-zero entries and all spaces finite-dimensional.
Assume that $V$ and $W$ are interleaved with respect to two nondecreasing functions $\psi_1,\psi_2$, i.e., there exist commuting linear maps $\phi_{W,V}^{\psi_1(t),t}:V(t)\to W(\psi_1(t))$ and $\phi_{V,W}^{\psi_2(t),t}:W(t)\to V(\psi_2(t))$. 

We consider the persistence diagrams $P_V$ and $P_W$ of $V$ and $W$ as collections in $\ZZ^2$ with multiplicity.
We say that:
\begin{enumerate}
	\item Two pairs $(b,d]_V\in P_V$ and $(\hat b,\hat d]_W\in P_W$ are \emph{related}, written as $(b,d]_V\sim (\hat b,\hat d]_W$, if $\hat b< \psi_1(b+1)$, $b< \psi_2(\hat b+1)$, $\psi_2(\hat d)\ge d$, and $\hat d\le \psi_1(d)$ all hold.
	\item A pair $(b,d]_V\in P_V$ is called alive if $\psi_1(\psi_2(b+1))\le d$. Likewise, a pair $(\hat b,\hat d]_W\in P_W$ is called alive if $\psi_2(\psi_1(\hat b+1))\le \hat d$.
\end{enumerate}
We will show how to construct a matching between $P_{V}$ and $P_{W}$ that only matches related entries and uses all living entries in $P_V$ and $P_W$ simultaneously.

First, we construct a partially defined injective map $M_{W,V}:P_V\to P_W$, such that
\begin{enumerate}
	\item all living $x\in P_V$ are matched, i.e., $M_{W,V}(x)\in P_W$ is well-defined,
	\item we only match related pairs, i.e., $M_{W,V}(x)\sim x$.
\end{enumerate}
This is done by the following procedure, where we construct the matching iteratively over increasing $b$, and adjust the persistence modules $V,W$ in every step:
\begin{enumerate}
	\item Suppose the iteration has arrived at $V(b)=\{0\}\neq V(b+1)$. If instead $V(b)=\{0\}$ for all $b\in\ZZ$, then we are done.
	\item Choose any $\xi=\xi^{d,b}_{b+1,i}\in V(b+1)$, and set $\zeta=\phi_{W,V}^{\psi_1(b+1), b+1}\xi \in W(\psi_1(b+1))$.
	\item If $\zeta=0$, jump to step $5$.
	\item Let $\hat\zeta\in W(\hat b+1)$ an earliest preimage of $\zeta$ in $W$ with lifetime $(\hat b, \hat d]$, i.e.,
	\begin{align*}
	\hat b&= \argmin_{\hat b}\{\exists \hat\zeta \in W(\hat b+1)\ \textrm{such that }\phi_{W,W}^{\psi_1(b+1), \hat b+1}\hat\zeta=\zeta\},\\
	\hat d&= \argmin_{\hat d}\{\phi_{W,W}^{\hat d+1,\psi_{1}(b+1)}\zeta = 0\},
	\end{align*}
	and add the edge $(b,d]_V\to (\hat b,\hat d]_W$ to the matching $M_{W,V}$.
	\item We update $W$ and $V$ by by factoring out the images of $\xi$ and $\zeta$, i.e., by going to
	\begin{align*}
	\widetilde V(k)&\colonequals V(k) / \vspan\{\xi^{d,b}_{k,i}\} &&\textrm{ for }b< k\le d,\\
	\widetilde W(k)&\colonequals W(k) / \vspan\{\phi_{W,W}^{k,\psi_1(b+1)}\xi^{d,b}_{b+1,i}\} &&\textrm{ for }\psi_1(b+1)\le k\le \hat d.
	\end{align*}
	It is clear that we still get have interleaved persistence diagrams $P_{\widetilde V}$ and $P_{\widetilde W}$, i.e., all the maps on the factor spaces are well-defined commute.
	
	$P_{\widetilde V}$ is obtained from $P_{V}$ by removing one $(b,d]_V$-entry. $P_{\widetilde W}$ is obtained from $P_W$ by removing the matched $(\hat b, \hat d]_W$-entry, and inserting a $(\hat b, \psi_1(b+1)-1]_W$-entry. This new entry is spurious; but, since we iterate in increasing $b$, it will never be used in the matching. We continue with Step $2$.
\end{enumerate}
The obtained map $M_{W,V}$ has indeed the desired properties:
\begin{enumerate}
	\item If $\xi$ is alive, i.e., $\psi_1(\psi_2(b+1))\le d$, then
	$\zeta\neq 0$, because
	\begin{align*}
	\phi_{V,W}^{\psi_2(\psi_1(b+1)),\psi_1(b+1)}\zeta &= \phi_{V,W}^{\psi_2(\psi_1(b+1)),\psi_1(b+1)} \phi_{W,V}^{\psi_1(b+1),b+1}\xi\\
	& = \phi_{V,V}^{\psi_2(\psi_1(b+1),b+1}\xi \neq 0.
	\end{align*}
	\item It holds that $(b,d]_V\sim (\hat b,\hat d]_W$:
	\begin{enumerate}
		\item Obviously $\hat b < \psi_1(b+1)$.
		\item Suppose that $b\ge \psi_2(\hat b+1)$. Then we can set \[\phi_{V,W}^{b,\hat b + 1}\colonequals \phi_{V,V}^{b,\psi_2(\hat b+1)}\circ  \phi_{V,W}^{\psi_2(\hat b+1), \hat b + 1}: W(\hat b +1)\to V(b),\] and this commutes such that
		\begin{align*}
		\phi_{W,V}^{\psi_1(b+1),b+1}\phi_{V,V}^{b+1,b}\phi_{V,W}^{b,\hat b+1}\hat\zeta= \phi_{W,W}^{\psi_1(b+1), \hat b+1}\hat\zeta = \zeta.
		\end{align*}
		As $\zeta\neq 0$, this contradicts the induction assumption $V(b)=\{0\}$.
		\item Suppose that $\hat d \ge \psi_1(d+1)$. This cannot be, since 
		\[\phi_{W,W}^{\hat d, \psi_1(d+1)}\phi_{W,V}^{\psi_1(d+1),d+1}\phi_{V,V}^{d+1,b+1}\xi=0=\phi_{W,W}^{\hat d,\psi_1(b+1)}\zeta.\]
		\item Suppose that $d \ge \psi_2(\hat d+1)$. Then $\phi_{V,V}^{\psi_2(\hat d+1),b+1}\xi\neq 0$.  This cannot be, since by construction $\hat d+1\ge\psi_1(b+1)$, and 
		\begin{align*}
		\phi_{V,V}^{\psi_2(\hat d+1),b+1}\xi
		&=\phi_{V,W}^{\psi_2(\hat d+1),\hat d+1}\phi_{W,W}^{\hat d+1,\psi_1(b+1)}\zeta
		=0.
		\end{align*}
	\end{enumerate}
\end{enumerate}

Using the same procedure, we can construct $M_{V,W}: P_{W}\to P_{V}$ that only matches related entries and is defined for all living $(\hat b,\hat d]_W\in P_W$.

Finally, we can construct the matching $M$:
We have a directed bipartite graph $M_{W,V}\cup M_{V,W}$. Every entry has at most one outgoing edge and at most one incoming edge. 
This means that this graph decomposes into directed pathes and cycles. For each cycle, we add every second edge to $M$ (cycles are even). For each path, we add every second edge, starting with the first one in directed order.

\end{document}